\documentclass[a4paper,11pt,reqno]{amsart}
\usepackage{graphicx}
\usepackage{amsmath}
\usepackage[T1]{fontenc}
\usepackage{amssymb}
\usepackage{amsthm}
\usepackage{color}
\usepackage{xcolor}
\usepackage[lmargin=2.5 cm,rmargin=2.5 cm,tmargin=3.5cm,bmargin=2.5cm,
paper=a4paper]{geometry}

\usepackage[english]{babel}

\newcommand{\Ab}{\mathbf A}

\newcommand{\Fb}{\mathbf F}

\newcommand{\R}{\mathbb R}

\newcommand{\C}{\mathbb C}

\newcommand{\dsp}{\displaystyle}

\DeclareMathOperator{\IM}{Im}
\DeclareMathOperator{\curl}{curl}
\DeclareMathOperator{\Div}{div}

\DeclareMathOperator{\dist}{dist}
 
\newtheorem{thm}{Theorem}[section]
\newtheorem{prop}[thm]{Proposition}

\newtheorem{theorem}[thm]{Theorem}

\theoremstyle{remark}

\numberwithin{equation}{section}
\title[2D Ginzburg-Landau  functional]{Breakdown of superconductivity in a magnetic field with self-intersecting zero set}
\author[K. Attar]{}
\author[]{Kamel Attar}
\begin{document}
	\everymath{\displaystyle}
\begin{abstract}
We prove that the lowest eigenvalue of the Laplace operator with a magnetic field having a self-intersecting zero set is a monotone function of the parameter defining the strength of the magnetic field, in a neighborhood of infinity. We apply this monotonicity result on the study of the transition from superconducting to normal states for the Ginzburg-Landau model, and prove that  the transition occurs at a unique threshold value of the applied magnetic field.
\end{abstract}
\maketitle 
\section{Introduction}
\subsection{Breakdown of superconductivity.}
Superconductivity is a state of metals and alloys that appear below a certain critical temperature $T_c$. When the temperature is below $T_c$ the body must be in the superconducting state; if an applied magnetic field (with small intensity) attempts to destroy the superconductivity, an induced magnetic field appears and repels the applied magnetic field. If the intensity of the applied magnetic field is increased gradually past a critical value, the superconductor can no longer resist the magnetic field and the appearance of superconductivity in the metal decreases gradually past critical conducting states (superconducting state, mixed conducting state and normal conducting state). Eventually, a large magnetic field destroys superconductivity from the sample (see Figure \ref{df-states}).

\begin{figure}[ht!]
	\begin{center}
		\includegraphics[scale=1]{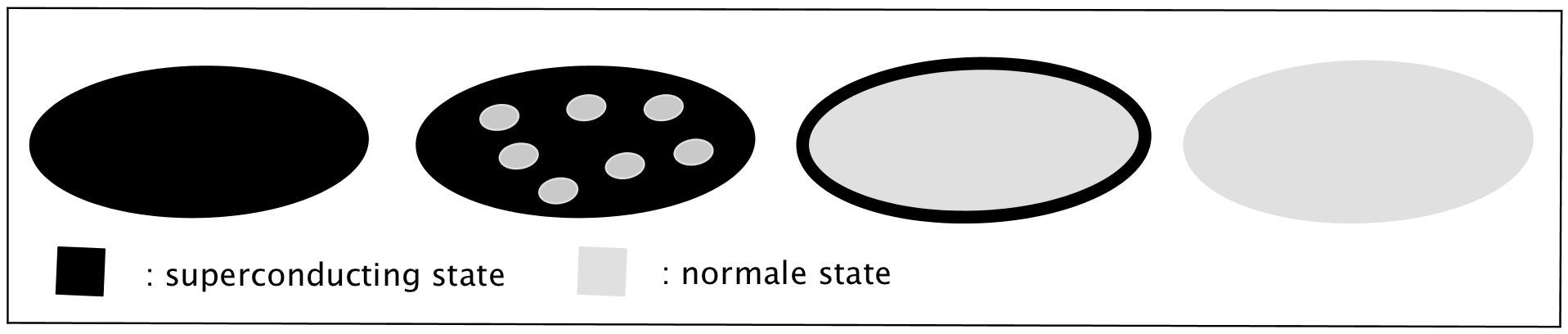}
		\caption{Schematic representation of a material shows the difference states.}
		\label{df-states}
	\end{center}
\end{figure}
In the case of a uniform applied magnetic field, phase transitions associated with a type II superconductor are marked by three thresholds:
\begin{itemize}
	\item $H_{C_1}$ is the intensity of the magnetic field where a superconductor switches from the superconducting state to the mixed conducting state (vortex state). In this case the superconductor allows the applied magnetic field to pass through small regions of the sample \cite{SaSe}.
	\item $H_{C_2}$ is the intensity of the  magnetic field that corresponds to the transition from superconductivity located in regions far from the boundary to surface superconductivity \cite{Pan, SaSe2, FK2}.
	\item  $H_{C_3}$ is the intensity of the magnetic field that corresponds to the transition from surface superconductivity to the normal conducting state; in this case the superconductivity can no longer exist in the superconductor \cite{Pan, FH3, CR}.
\end{itemize}
Non uniform magnetic fields is the subject of recent research:
\begin{itemize}
\item Non-vanishing smooth magnetic fields \cite{FH3, Ray, LP}\,;
\item Sign-changing smooth magnetic fields \cite{PaKw, HK}\,;
\item Non-smooth magnetic fields  \cite{KA1, AsKaSu, KHWE}\,.
\end{itemize}
In this contribution, we study smooth magnetic fields with a self-intersecting zero set.

\subsection{Magnetic field with self-intersecting zeros} Suppose that $\Omega\subset\R^2$ is bounded, open, simply connected with smooth boundary. Consider a vertical magnetic field of the form
\[B_0\vec{z} \]
where the function $B_{0}\in C^{\infty}(\overline{\Omega})$ has a non-trivial zero set
\begin{equation}\label{gamma}
\Gamma=\{x\in\overline{\Omega}: B_{0}(x)=0\}\qquad\text{and}\qquad \Sigma=\{x\in\Gamma: \nabla B_{0}(x)=0\}
\end{equation}
and satisfies  the following properties (see  \cite{MJN})
\[\Gamma\neq \varnothing,\quad \Sigma\neq\varnothing{\rm ~is ~finite,}\quad \Sigma\cap \partial\Omega= \varnothing \]
and
\begin{equation}\label{B(x)}
\left\{
\begin{array}{rll}
|B_{0}| + |\nabla B_0 | >0&\mbox{ in } \overline{\Omega}\setminus \Sigma& (a)\\
\nabla B_{0}\times\vec{n}\neq 0 &\mbox{ on } \Gamma\cap\partial\Omega&(b)\\
|\nabla B_{0}| + |{\rm Hess} B_0 | >0&\mbox{ on } \Gamma& (c)\\
\vec{\tau_{1}}\times \vec{\tau_{2}}\neq 0 &\mbox{ on } \Sigma&(d)\,.
\end{array}
\right.
\end{equation}
Here $ \vec{n} $ is the unit normal vector to the boundary, $ \vec{\tau}_{1}$ and  $ \vec{\tau}_{2} $ are the unit tangent vector on the intersection point $ x\in\Sigma $, and $ {\rm Hess} B_0 $ is the Hessian matrix of the magnetic field at the point $ x $ which has two non-zero eigenvalues $ \lambda^{\rm Hess}_{1}(x) $ and $ \lambda^{\rm Hess}_{2}(x) $ with opposite signs  and labeled as follows
\begin{equation}\label{eq:comp-lambda-Hess}
|\lambda^{\rm Hess}_{1}(x)|\leq |\lambda^{\rm Hess}_{2}(x)|\quad(x\in\Sigma)\,.
\end{equation}
Assumption $ \eqref{B(x)}_{(a)} $ implies that for any open set $\omega$ relatively compact in $\Omega$, $\Gamma\cap\omega$ is  either empty, or consists of a union of smooth curves and the quantity $ |\nabla B_0| $ does not vanish on $ \Gamma $. We see from $ \eqref{B(x)}_{(b)} $ that the $B_0$ is allowed to vanish on a finite number of boundary points and that the curve $ \Gamma $ cannot intersect tangentially $ \partial\Omega$. Assumptions $ \eqref{B(x)}_{(c)} $ and $ \eqref{B(x)}_{(d)} $ tell us that the smooth curves can intersect on isolated points where $ |\nabla B_0| $ vanishes.

The case where the magnetic field vanishes non-degenerately along a smooth  non self intersecting  curve, is the subject of numerous works in the contexts of shape optimization\cite{M}, superconductivity \cite{PaKw, HK, KA3} and   semiclassical spectral asymptotics \cite{DR, BNR}.

\begin{figure}[ht!]
	\begin{center}
		\includegraphics[scale=0.4]{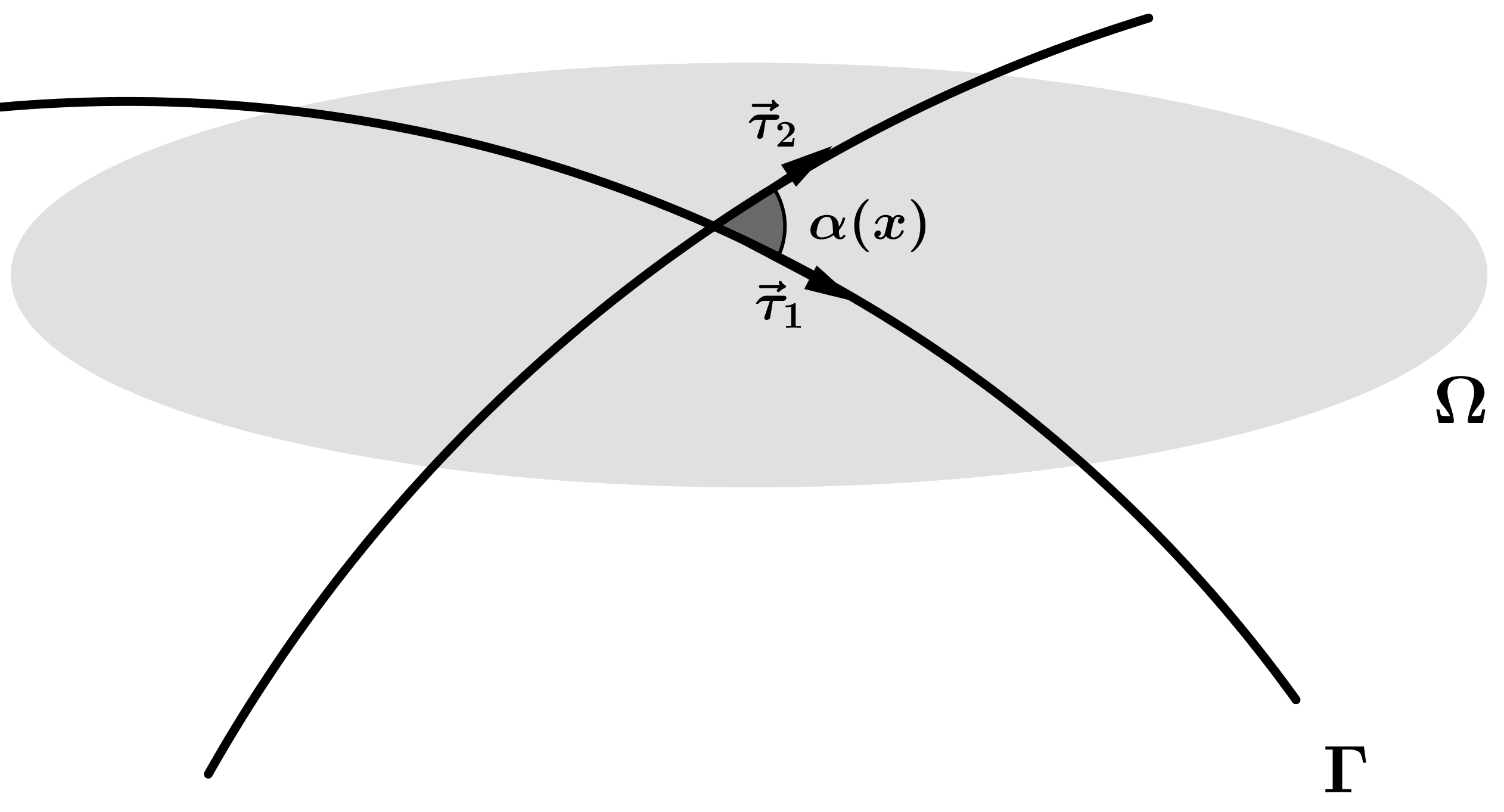} 
		\caption{Schematic representation of a domain subject to a variable magnetic field that vanishes along the curve $\Gamma$ which crosses itself inside the domain.}
		\label{domain1}
	\end{center}
\end{figure}

\subsection{Magnetic Laplacian and strong diamagnetism}

Consider the unique vector field $\Fb\in H^1(\Omega;\R^2)$  satisfying the following properties
\begin{equation}\label{eq:def-F}
\nu\cdot\Fb=0~{\rm on~}\partial\Omega, \quad {\rm div}\Fb=0~\&~\curl\Fb=B_0{~\rm in~}\Omega\,,
\end{equation}
where $\nu$ is the unit outward normal of the boundary of $\Omega$.

Let us introduce the  magnetic Schr\"odinger operator 
\begin{equation}\label{def:P}
	P_{b\Fb }^{\Omega}=-(\nabla-ib\Fb)^{2}\,,
\end{equation}
with domain
\[
D(P_{b\Fb }^{\Omega}):=\{u\in\mathcal{V}(\Omega)~:~ P_{b\Fb}^{\Omega}u\in L^{2}(\Omega)~\&~ \nu\cdot(\nabla-i b\Fb)u=0~{\rm on}~\partial \Omega\}\,.
\]
The lowest eigenvalue (ground state energy) of this operator is 
\begin{equation}\label{def:mu1}
	\mu_{1}(b \Fb)=\inf_{\substack{\phi\in H^{1}(\Omega)\\ \phi\neq 0}}\left(\frac{\mathcal{Q}_{b \Fb}^{\Omega}(\phi)}{\|\phi\|^{2}_{L^{2}(\Omega)}}\right)\,,
\end{equation}
where
\begin{equation}\label{Quad}
	\mathcal{Q}_{b \Fb}^{\Omega}(\phi)=\int_{\Omega}|(\nabla-ib \Fb)\phi|^{2}\,dx\,.
\end{equation}
The large field limit $b\to+\infty$  can be transformed to a semi-classical limit by introducing the small  parameter $h=\frac1b$. The full asymptotic expansion of the lowest eigenvalue $P_{b\Fb}^\Omega$ is derived by Dauge, Miqueu and Raymond \cite{MJN}.
Based on it, we prove the monotonicity of $ \mu_{1}(b \Fb) $ with respect to $b$ in a neighborhood of $+\infty$. This property is named \emph{strong diamagnetism} in the literature \cite{FH3, FH, FH0}.

\begin{theorem}\label{main-th} There exists $ b _0>0$, such that the function $b \longmapsto \mu_{1}(b \Fb)$ is monotone increasing on $ [b _0, +\infty [ $.
\end{theorem}

Strong diamagnetism has been proved earlier under different assumptions on $ \Omega $ and $ B_0 $  \cite{FH,BF,  WA}. Counterexamples of strong diamagnetism do exist   \cite{FP, HK2,  KHPA2, KHSU}.

\subsection{Ginzburg-Landau model}

Consider the Ginzburg-Landau functional
\begin{equation}\label{eq-2D-GLf}
	\mathcal E_{\kappa,H,B_{0}}(\psi,\Ab)= \int_\Omega\left( |(\nabla-i\kappa
	H\Ab)\psi|^2+\frac{\kappa^2}{2}(1-|\psi|^2)^2\right)\,dx
	+\kappa^2H^2\int_{\Omega}|\curl\Ab-B_0|^2\,dx\,.
\end{equation}
The vector field $\Ab$ is called the magnetic potential which  describes the induced magnetic field in the sample via
$$\curl \Ab=\overrightarrow{\nabla}\times \Ab\,.$$  
The positive number $\kappa$ is the characteristic scale of the superconductor that distinguishes the material of the sample; we will study the case of materials of type II where $\kappa$ is sufficiently large ($\kappa\to\infty$). The complex-valued function $\psi$ is the order parameter that rises in the superconductor phase, where the modulus squared $|\psi(x)|^2$ can be interpreted as the local density of the superconducting electron Cooper pairs as follows:
\begin{itemize}
	\item If $|\psi(x)|\neq 0$, then at location $x$, the material is in the superconducting state.
	\item If $|\psi(x)|=0$, then at location $x$, there exists no Cooper pairs, and the material is in the normal state.
	\item If $|\psi(x)|$ has zeros but does not vanish identically, the material is in the {\it mixed state}.
\end{itemize}
Since the functional is invariant under gauge transformations $(\psi,\Ab)\to (e^{i\kappa H\chi}\psi, \Ab-\nabla \chi)$, it is enough to consider configurations $(\psi,\Ab)$ in the space $H^1(\Omega;\C)\times H^1_{\Div}(\Omega)$ of sobolev functions, where
\begin{equation}\label{eq-2D-hs} H^1_{\Div}(\Omega)=\Big\{\Ab=(\Ab_{1},\Ab_{2})\in
	H^1(\Omega)\times H^1(\Omega) ~:~\Div \Ab=0~{\rm in}~\Omega \,,\,\Ab\cdot\nu=0~{\rm
		on}\,
	\partial\Omega \,\Big\}\,,
\end{equation}
with  $\nu$ being the unit interior normal vector of
$\partial\Omega$.

The equilibrium state of the system will be where the total energy $\mathcal E_{\kappa,H,B_{0}}$ is minimal, we denote by $(\psi,\Ab)$ the minimizing configuration of $\mathcal E_{\kappa,H,B_{0}}$. When $\psi=0$, the minimizing configuration is in the form $(0,\Fb)$ with $\Fb$ in $H^1_{\Div} (\Omega)$  such that $\curl\Fb=B_{0}$. Notice that this solution is unique (up to
a gauge transformation) and we call the pair $(0,\Fb)$ the normal state.

\subsection*{Ginzburg-Landau equations}
Minimizers, $(\psi,\Ab)\in H^1(\Omega;\C)\times H^1_{\Div}(\Omega)$ of \eqref{eq-2D-GLf} have to satisfy the Ginzburg-Landau equations,
\begin{equation}\label{eq-2D-GLeq}
	\left\{
	\begin{array}{llll}
		-(\nabla-i\kappa H\Ab)^2\psi=\kappa^2\, (1-|\psi|^2)\, \psi&{\rm in}&
		\Omega&(a)
		\\
		-\nabla^{\bot}\curl(\Ab-\Fb)=\displaystyle\frac1{\kappa
			H}\IM(\overline{\psi}\,(\nabla-i\kappa
		H\Ab)\psi) &{\rm in}&\Omega&(b)\\
		\nu\cdot(\nabla-i\kappa H\Ab)\psi=0&{\rm
			on}&\partial\Omega&(c)\\
		\curl\Ab=\curl\Fb&{\rm on}&\partial\Omega&(d)\,.
	\end{array}\right.
\end{equation}
Here,
$\nabla^{\bot}\curl\Ab=\Big(\partial_{x_2}(\curl\Ab),
-\partial_{x_1}(\curl\Ab)\Big)$. The boundary conditions state that the supperconducting current can not flow outside  the sample (the sample is surrounded by vacuum).

A  solution   $(\psi,\Ab)$ of \eqref{eq-2D-GLeq} is said to be trivial if $\psi=0$ everywhere (consequently $\curl \Ab=b $ everywhere). Let us recall a simple criterion \cite[Lemma 2.1]{LP} for the existence of a non-trivial solution of the Ginzburg-Landau equations in \eqref{eq-2D-GLeq}.

\begin{prop}[Sufficient condition for existence of non-trivial solutions]\label{prop:LP}
	For all $ \kappa $, $ H>0 $, if $\mu_{1}(\kappa H\Fb)<\kappa^2$, then every minimizer of $\mathcal E_{\kappa,H,B_{0}}$ is non-trivial.
\end{prop}

Our main result proves that the condition $\mu_{1}(\kappa H\Fb)<\kappa^2$ is also necessary, provided $\kappa$ is sufficiently large.

\begin{thm}[Necessary condition for existence of non-trivial solutions]\label{thm:mu<k}
There exists $\kappa_0>0$ such that, for all $\kappa\geq \kappa_0$ and $H>0$, the following two properties are equivalent:
\begin{itemize}
\item[A.] There exists a solution $(\psi,\Ab)$ of \eqref{eq-2D-GLeq} such that $\psi\not\equiv0$\,.
\item[B.] $\mu_{1}(\kappa H\Fb)<\kappa^2$.
\end{itemize}
\end{thm}

For non sign-changing smooth magnetic fields and smooth domains, Theorem~\ref{thm:mu<k} was proved by Fournais-Helffer \cite{FH1}. For uniform magnetic fields and non-smooth domains, it was proved by \cite{BF}.

\subsection*{Discussion}

Combining Theorems~\ref{main-th} and Theorem~\ref{thm:mu<k}, we know that   the  transition from trivial to non-trivial solution occurs at a unique threshold $H_{C_3}(\kappa)$.  More precisely, 
\[\exists\,\kappa_0>0,~\forall\,\kappa\geq\kappa_0,~\exists\,H_{C_3}(\kappa)>0\]
such that
\begin{itemize}
\item $H=H_{C_3}(\kappa)$ is the unique solution of $\mu_1(\kappa H)=\kappa^2$\,;
\item For $H<H_{C_3}(\kappa)$, every minimizer of $\mathcal E_{\kappa,H,B_{0}}$ is non-trivial (i.e. the order parameter does not vanish everywhere)\,;
\item For $H\geq H_{C_3}(\kappa)$, every solution $(\psi,\Ab)$  of \eqref{eq-2D-GLeq} is trivial (i.e. $\psi\equiv0$)\,.
\end{itemize}

The existing spectral asymptotics \cite{MJN} for the eigenvalue $\mu_1(b\Fb)$ also yield a complete asymptotic expansion of $H_{C_3}(\kappa)$; in fact (see Theorem~\ref{thm.mu1.ful})
\[ H_{C_3}(\kappa)\sim\kappa^3\sup_{x\in  \Sigma}\left(\frac{1}{|\lambda^{\rm Hess}_{2}(x)|\left [\mathcal{X}_{1}\Big(\varepsilon(x)\Big)\right ]^2}\right )+ \sum\limits_{j=1}^{+\infty} c_j\kappa^{3-j}\qquad (\kappa\to+\infty)\,.\]
Here
\begin{equation}\label{eps}
	\varepsilon(x)=\sqrt{\frac{|\lambda^{\rm Hess}_{1}(x)|}{|\lambda^{\rm Hess}_{2}(x)|}},
\end{equation}
and for $\varepsilon>0$,  $ \mathcal{X}_{1}(\varepsilon) $ denotes the first eigenvalue of the following operator
\begin{equation}\label{Model.op}
	\mathcal{J}^{\varepsilon}:= D_{t}^2 +\left (D_s- \left (\frac{t^3}{3}-\varepsilon^2 s^2 t\right ) \right  )^2\,, \qquad \text{in}~L^2(\R^2)\,.
\end{equation}

Furthermore,  for $\kappa\geq\kappa_0$ and $H<H_{C_3}(\kappa)$, the local minimizers of the functional $\mathcal E_{\kappa,H,B_{0}}$ are necessarily non-trivial. In  fact, $(0,\Fb)$ is a local minimizer (locally stable) provided that 
the second variation of the  $\mathcal E$ is positive,
$$
{\rm Hess}_{(0,\Fb)} >0\,,
$$
which is equivalent to the property $\mu_1(\kappa H\Fb)> \kappa^2$. 

\subsection*{Oscillations and Little-Parks effect} 

The case studied in this paper is consistent with the generically observed  monotone transition between superconduting and normal phases. Equally interesting are  non-generic cases where the sample oscillates between the superconducting and normal phases before setting definitely to the normal state (Little-Parks effect \cite{LiPa},  see Fig.~\ref{domain}). Breaking the monotonic transition is viable by  topological obstructions related to  the domain or the magnetic field. We refer to \cite{FP, KHSU, HK2, KHPA, KHPA2} for settings with oscillations and to \cite{FH1, BF, FK1, WA} for generic monotone settings. 

\begin{figure}[ht!]
	\begin{center}
		\includegraphics[scale=1.7]{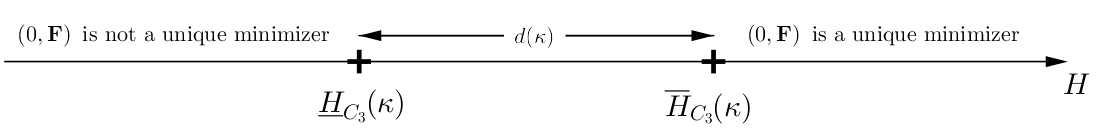} 
		\caption{Schematic representation of a sample consistent with the Little-Parks effect: Between $\overline{H}_{C_3}(\kappa)$ and $\underline{H}_{C_3}(\kappa)$, the sample oscillates between normal and superconducting phases.}
		\label{domain}
	\end{center}
\end{figure} 

\subsection*{Organization of the paper}
The rest of the paper consists of three sections: Sec.~\ref{sec:sd} is devoted to the analysis of the monotonicity of $ b \to  \mu_{1}(b\Fb)  $; Sec.~\ref{GP} establishes the breakdown of superconductivity under Assumption \ref{B(x)}; finally, Sec.~\ref{CBS} is devoted to the proof of Theorem~\ref{thm:mu<k} on the superconducting-normal phase transition.

\section{Strong diamagnetism}\label{sec:sd}

In this section we prove Theorem~\ref{main-th} regarding the  monotonicity of the function $ b  \mapsto  \mu_{1}(b \Fb)  $. We follow the approach of Fournais-Helffer \cite{FH} whose ingredients are
\begin{itemize}
\item The full asymptotic expansion for the lowest eigenvalue\,;
\item Proposition~\ref{pro.fh} below.
\end{itemize} 

The full asymptotic expansion for $\mu_1(b\Fb)$ was obtained in \cite{MJN}; we recall this result below.

	\begin{theorem}[Dauge-Miqueu-Raymond \cite{MJN}]\label{thm.mu1.ful}
		Let $\Omega\subset\R^2$ be an open bounded set with a smooth boundary. Under Assumption \ref{B(x)}. There exists a sequence $ \{\alpha_j \}_{j=0}^{\infty} \subset \R^+$ such that for all $M \geq 0$,	
		\begin{equation}\label{ful.exp.mu1}
			\mu_{1}(b \Fb)=b^{\frac{1}{2}} \sum_{j=0}^{M} \alpha_j  b^{\frac{-j}{4}} +\mathcal{O}\left (b^{\frac{1-M}{4}}\right )	\,.
		\end{equation}
	Here, \begin{equation}\label{m_eig}
		\dsp \alpha_0= \Lambda_{1}=\min_{x\in \Sigma }|\lambda^{\rm Hess}_{2}(x)|^{1/2} \mathcal{X}_{1}\Big(\varepsilon(x)\Big)\,,
	\end{equation}
$\varepsilon(x)  $ and $  \mathcal{X}_{1} $  are introduced in \eqref{eps} and \eqref{Model.op} respectively.
\end{theorem}

Recall the definition of the left and right derivatives of $ \mu_{1}( b \Fb)$, 
$$ \frac{d}{d b }\mu^{\pm}_{1}( b \Fb)=\lim_{\varepsilon\to 0^{\pm} } \frac{\mu_{1}( (b+\varepsilon) \Fb)-\mu_{1}( b \Fb)}{\varepsilon}\,,$$
whose  existence  is guaranteed by the analytic perturbation theory \cite{TK}, which also gives for all $  b >0 $ 
\begin{equation}\label{com.Le.Ri}
	\frac{d}{d b }\mu_1^{+}( b \Fb)\leq \frac{d}{d b }\mu_1^{-}( b \Fb)\,.
\end{equation}
We introduce the new parameter $\beta=\sqrt{b}$ and the following function
\begin{equation}\label{eq:new-ev}
\lambda(\beta)=\mu_1( \beta^2 \Fb)\,.
\end{equation}
Recall the following proposition \cite[Page 27]{FH} 
\begin{prop}\label{pro.fh}
Let $ g $ be a function such that for all $\varepsilon \in (-1,\, 1)$ we have
\begin{equation}\label{cond.g}
	|g(\beta+\varepsilon)-g(\beta)|\to 0\,,\qquad \text{as}~ \beta\to +\infty\,.
\end{equation}
Suppose that there exists $\alpha\in \R$ such that 
\[
	\lambda( \beta)=\alpha \beta+g(\beta)+o(1)\,,\qquad \text{as}~ \beta\to +\infty\,.
\]
Then the limits $ \lim_{\beta\to+ \infty}	\frac{d}{d \beta }\lambda(\beta) $ and $ \lim_{\beta\to +\infty}	\frac{d}{d \beta }\lambda(\beta \Fb) $ exist and
\[
	\lim_{\beta\to +\infty}	\frac{d}{d \beta }\lambda^+(\beta) =\lim_{\beta\to +\infty}	\frac{d}{d \beta }\lambda^-(\beta)=\alpha\,. 
\]
\end{prop}
\begin{proof}[Proof of Theorem~\ref{main-th}]
Choose $M=3$ in Theorem~\ref{thm.mu1.ful}. The proof of Theorem \ref{main-th} simply follows by rewriting \eqref{ful.exp.mu1} in the following form (recall that $\beta=\sqrt{b}$),
\[\lambda(\beta)=\alpha_0\beta+g(\beta)+ o(1)\,,\]
where $ g(\beta)=\beta \sum_{j=1}^{3} \alpha_j  \beta^{\frac{-j}{2}} $ and $\lambda(\beta)$ is introduced in \eqref{eq:new-ev}. For  $\varepsilon \in (-1,\, 1)$ fixed, we have
\[
g(\beta+\varepsilon)-g(\beta)\underset{\beta\to+\infty}{\sim}\frac{\alpha_1\varepsilon}{(\beta+\varepsilon)^{\frac{1}{2}} +\beta^{\frac{1}{2}} }\]
which implies that \eqref{cond.g} holds.
We can use Proposition \ref{pro.fh} in order to prove the monotonicity of $ \mu_{1}(b \Fb) $. In fact,
\[
\lim_{\beta\to +\infty}	\frac{d}{d \beta }\mu^+_{1}(\beta^2 \Fb) =\lim_{\beta\to +\infty}	\frac{d}{d \beta } \mu^{-}_{1}(\beta^2 \Fb)=\alpha_0\,. 
\]
To finish the proof, we write by the chain rule (recall that $\beta=\sqrt{b}$),
$$ 
		\lim_{b\to +\infty}	b^{\frac{1}{2}}\frac{d}{d b }\mu_1^{+}( b \Fb) =\lim_{b\to +\infty}	b^{\frac{1}{2}}\frac{d}{d b }\mu_1^{-}( b \Fb)=\frac{1}{2}\alpha_0\,. 
$$
\end{proof}

\section{Breakdown of superconductivity}\label{GP}
In this section we extend a result of Giorgi-Phillips \cite{GP} to the case where the zero set of $ B_0 $ is self-intersecting. We will show that the trivial solution $ (0,\Fb) $ is a global minimizer, when $H/\kappa^3$ is large. We first recall {\it a priori} estimates for a critical point  $(\psi, \Ab)$ of the G.L.functional (see \cite[Theorem 8.1]{KA3}).

\begin{thm}\label{thm-2D-apriori}
There exist positive constants $C_{1}$ and $C_{2}$ such that, if $\kappa>0$, $H>0$ and $(\psi,\Ab)$ is a solution of \eqref{eq-2D-GLeq}, then,
\begin{align}
\|(\nabla-i\kappa H\Ab)\psi\|_{L^{2}(\Omega)}&\leq \,\kappa\, \|\psi\|_{L^{2}(\Omega)}\label{2nd-<}\\
\|(\Ab-\Fb)\psi\|_{L^{2}(\Omega)}& \leq \frac{C_{1}}{H}\|\psi\|_{L^{4}(\Omega)}^{2}\|\psi\|_{L^{2}(\Omega)}\label{5d-<}\\
\|(\nabla-i\kappa H\Fb)\psi\|_{L^{2}(\Omega)}&\leq C_{2}\,\kappa\,\|\psi\|_{L^{2}(\Omega)}\label{8d-<}\,.
\end{align}

\end{thm}

In the next theorem, we prove that $ (0, \Fb) $ is the unique minimizer of the functional when $ H/\kappa^3 $ is sufficiently large.

\begin{theorem}\label{thm:GP}
	 Let $\Omega\subset\R^2$ be a smooth, bounded, simply-connected open set. There exist positive constants $\bar C$ and $\kappa_{0}$, such that, if 
	 $$
	 H\geq \bar C \kappa^3\,,\qquad\kappa\geq\kappa_{0}\,,
	 $$
	then $(0,\Fb)$ is the unique solution to \eqref{eq-2D-GLeq}.
\end{theorem}

\begin{proof}
	We denote by $\mu_{1}^{N}( b \Fb)$ the lowest eigenvalue of the $\rm Schr\ddot{o}dinger$ operator $P_{ b \Fb}^{\Omega}$ with Neumann condition in $L^{2}(\Omega)$. Assume that we have a \textbf{non normal} critical point $(\psi, \Ab)$ for $\mathcal E_{\kappa,H,B_{0}}$. This means that $(\psi, \Ab)\in H^{1}(\Omega)\times H^{1}_{{\rm div}}(\Omega)$ is a solution of \eqref{eq-2D-GLeq} and
	\begin{equation*}\label{psi>0}
		\int_{\Omega}|\psi|^{2}\,dx>0\,.
	\end{equation*}
	Using \eqref{2nd-<} with $ b =\kappa H$, we get
	$$
	\frac{\|(\nabla-i b \Fb)\psi\|_{L^{2}(\Omega)}^{2}}{\|\psi\|^{2}_{L^{2}(\Omega)}}\leq  \kappa^{2}\,,
	$$
	which implies by assumption that the lowest Neumann eigenvalue satisfies,
	\begin{equation*}\label{muN<}
		\mu_{1}^{N}( b \Fb)\leq  \kappa^{2}\,.
	\end{equation*}
Using the leading order asymptotics in \eqref{ful.exp.mu1} with $ b =\kappa H$ and $ H\geq \bar C\kappa^3 $, we obtain
\begin{equation}\label{MJN}
\Lambda_{1}=	\lim_{\kappa H\longrightarrow+\infty}\frac{\mu_{1}^{N}(\kappa H\Fb)}{\sqrt{\kappa H}}\leq \bar C^{-\frac{1}{2}}\,.
\end{equation}
For $ \bar C>(\Lambda_1)^{-2} $, \eqref{MJN} yields a contradiction.
\end{proof}

\begin{figure}[ht!]
	\begin{center}
		\includegraphics[scale=1.5]{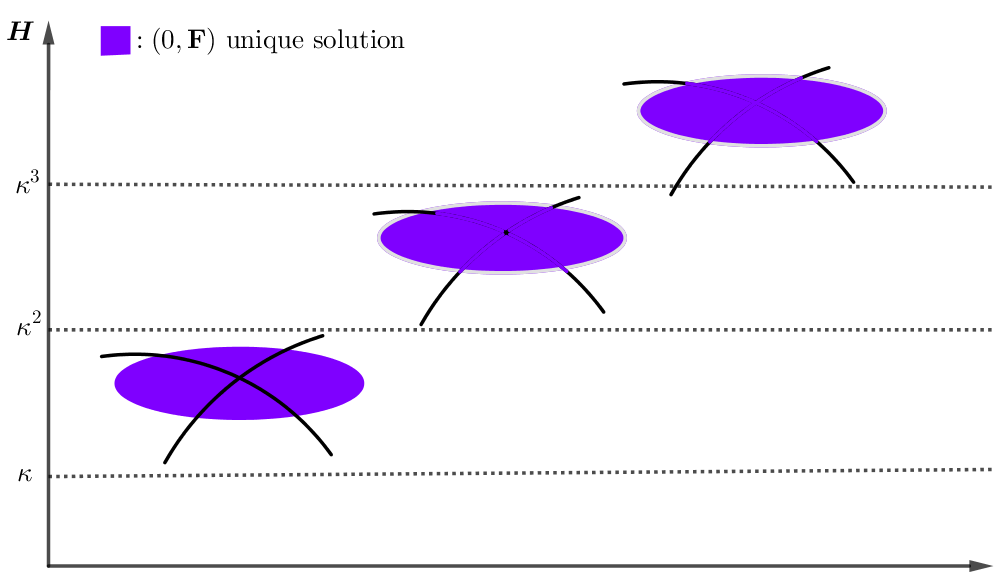} 
		\caption{Schematic representation of the levels of $ H $ that correspond to the trivial solution $(0,\Fb)$.}
		\label{Giorgi}
	\end{center}
\end{figure}

\section{Criterion for the breakdown of superconductivity}\label{CBS}

The purpose of this section is to prove Theorem~\ref{thm:mu<k}. 
One ingredient in the proof is 
 Proposition~\ref{prop:est-psi-var} below on the localization of the order parameter, which    is an adaptation of an analogous result obtained in \cite{HK} for the case where the magnetic field has a non-empty zero set but not self-intersecting. It says that, if $(\psi,\Ab)$ is a critical point of the functional  in \eqref{eq-2D-GLf} and  $H$ is of order $\kappa^3$, then $|\psi|$ is concentrated near the set $\Sigma=\{x\in\Omega~:~B_0(x)=0
 ,\&\,\nabla B_0(x)=0\}$.

 \begin{prop}\label{prop:est-psi-var}
 	Let $\varepsilon>0$ and $ \alpha\in \left (0,1\right ) $. There exist two positive constants $C$ and $\kappa_0$ such that, if
 	\begin{equation}\label{cond:H>kappa}
 		\kappa\geq\kappa_0\,,\quad H\geq\varepsilon\,\kappa^{3}\,,
 	\end{equation}
 	and $(\psi,\Ab)$ is a solution of \eqref{eq-2D-GLeq}, then
 	\begin{equation}\label{l3est}
 		\|\psi\|^{2}_{L^2(\Omega)}\leq C\,\kappa^{-\frac{\alpha}{2}}\|\psi\|^{2}_{L^4(\Omega)}\,.
 	\end{equation}
 \end{prop}

\begin{proof}
	Let $\ell=\kappa^{-\alpha}$ with $ 0<\alpha<1 $ and $$D_{\ell}=\big\{x\in\Omega:\dist(x,\partial\Omega)>\ell~\&~\dist(x,\Gamma)>\ell\big\}\,.$$ Let us introduce a cut off function $h\in C^{\infty}_{c}(\Omega)$ satisfying
	$$
	0\leq h \leq 1~{\rm in}~\Omega\,,\quad h=1~{\rm in}~D_{\ell}\,,\quad {\rm supp}\,h\subset D_{\ell/2}\quad{\rm and}\quad |\nabla h |\leq \frac{C}{\ell}\quad{\rm in}~\Omega\,,
	$$
	where $C$ is a positive constant.\\
	Adapting the proof by Helffer and Kachmar 
	\cite[Equation (6.6)]{HK}, there exists a positive constant $c$ such that,
	$$
	\kappa\,H\int_\Omega|B_{0}(x)|\,|h\psi|^2\,dx-c\,\kappa\,\|\psi\|_{L^{2}(\Omega)}\|h\psi\|_{L^{4}(\Omega)}^2\leq \int_{\Omega}|(\nabla-i\kappa H\Ab)h\psi|^2\,dx\,.
	$$
	Using Cauchy's inequality, we obtain
	$$
	c\,\kappa\,\|\psi\|_{L^{2}(\Omega)}\|h\psi\|_{L^{4}(\Omega)}^{2}\leq c^{2}\|\psi\|_{L^{2}(\Omega)}^{2}+\kappa^2\|h\psi\|_{L^{4}(\Omega)}^{4}\,,
	$$
	which implies that
	\begin{multline}\label{eq:22}
		\int_\Omega\,\left(\kappa\,H\,|B_{0}(x)|-\kappa^{2}\right)\,|h\psi|^2\,dx\leq\int_{\Omega}|(\nabla-i\kappa H\Ab)h\psi|^2\,dx-\kappa^{2}\|h\psi\|_{L^{2}(\Omega)}^{2}
	\\	+c^{2}\|\psi\|_{L^{2}(\Omega)}^{2}+\kappa^{2}\|h\psi\|_{L^{4}(\Omega)}^{4}\,.
	\end{multline}
	Multiplying both sides of \eqref{eq-2D-GLeq}$_a$ by $h^2\psi$, it results from an integration by parts that,
	\begin{equation}\label{eq:2}
	\int_{\Omega} \Big(|(\nabla-i\kappa H \Ab)h\psi|^{2}-\kappa^2 |h\psi|^{2}\Big)\,dx=	\int_{\Omega} |\nabla h|^{2}|\psi|^{2}\,dx-\kappa^{2}\int_{\Omega}h^{2}|\psi|^{4}\,dx\,.
\end{equation}
	Implementing \eqref{eq:2} into \eqref{eq:22}, we get,
	\begin{align*}
		\int_\Omega\,\left(\kappa\,H\,|B_{0}(x)|-\kappa^{2}\right)\,|h\psi|^2\,dx&\leq c^{2}\int_{\Omega}|\psi|^{2}\,dx+\int_{\Omega}|\nabla h|^{2}|\psi|^{2}\,dx+\kappa^{2}\int_{\Omega}(h^{4}-h^{2})|\psi|^{4}\,dx\\
		&\leq c^{2}\int_{\Omega}|\psi|^{2}\,dx+\int_{\Omega}|\nabla h|^{2}|\psi|^{2}\,dx\,.
	\end{align*}
	Here, we have used the fact that $h^{4}\leq h^{2}$ since $0\leq h\leq 1$.\\
	As a consequence of the non degeneracy of $ B_0 $ outside $ \Gamma $, there exists a positive constant $ C(B_0) $  such that
	\begin{equation}\label{eq.LB.nabla}
		|B_0(x)|\geq C(B_0)\min\Big( {\rm dist}(x,\Sigma)^2,{\rm dist}(x,\Gamma)\Big) \geq  C(B_0)\,\ell^2 \,,\qquad \forall x\in D_{\ell}\,.
	\end{equation}
	Therefore, by using \eqref{cond:H>kappa}, we get,
	$$
	\Big(C(B_0)\,\varepsilon\,\kappa^{4-2\alpha}-\kappa^{2}\Big)\int_\Omega|h\psi|^2\,dx\leq c^{2}\int_{\Omega}|\psi|^{2}\,dx+\int_{\Omega}|\nabla h|^{2}|\psi|^{2}\,dx\,.
	$$
	Writing \begin{equation}\label{dec.h}
		\displaystyle\int_{\Omega}|\psi|^{2}\,dx=\int_{\Omega}|h\psi|^{2}\,dx+\int_{\Omega}(1-h^{2})|\psi|^{2}\,dx
	\end{equation} 
    and using the assumption on $h$, we have,
	$$
	\Big(C(B_0)\,\varepsilon\,\kappa^{4-2\alpha}-\kappa^{2}-c^{2}\Big)\int_{\Omega}|h\psi|^{2}\,dx\leq \Big(c^{2}+C\,\kappa^{2\alpha}\Big)\int_{\Omega\setminus D_{\ell}}|\psi|^{2}\,dx\,.
	$$
	For $\kappa$ large enough, $\dsp
	C(B_0)\,\varepsilon\,\kappa^{4-2\alpha}-\kappa^{2}-c^{2}\geq \frac{1}{2}C(B_0)\,\varepsilon\,\kappa^{4-2\alpha}$ and
	$$
	\int_\Omega| h\psi|^{2}\,dx\leq 2\frac{C}{  C(B_0) \,\varepsilon}\,\kappa^{4(\alpha-1)} \int_{\Omega\setminus D_{\ell}}|\psi|^{2}\,dx \,.
	$$
	Using the decomposition in \eqref{dec.h} and that $ 1-h^2 $ vanishes outside $  D_\ell $,  we get further,
	$$
	\int_\Omega|\psi(x)|^{2}\,dx\leq \left(2\frac{C}{  C(B_0) \,\varepsilon}\,\kappa^{4(\alpha-1)}+1\right)\int_{\Omega\setminus D_{\ell}}|\psi|^{2}\,dx\leq  2\int_{\Omega\setminus D_{\ell}}|\psi|^{2}\,dx\,.
	$$
	Finally, by the Cauchy-Schwarz inequality,
	$$
	\int_{\Omega\setminus D_\ell}|\psi|^{2}\,dx\leq |\Omega\setminus D_\ell|^{1/2} \left(\int_{\Omega\setminus D_\ell} |\psi|^{4}\,dx\right)^{\frac{1}{2}}
	\leq C\,\kappa^{-\frac{\alpha}{2}}\left(\int_\Omega |\psi|^{4}\,dx\right)^{\frac{1}{2}}\,.
	$$
	This finishes the proof of the proposition.
\end{proof}

\begin{proof}[Proof of Theorem~\ref{thm:mu<k}]
Suppose that  $(\psi,\Ab)$ is a solution of \eqref{eq-2D-GLeq} with $\|\psi\|_{L^{2}(\Omega)}\neq 0$. Then $H\leq \bar C\kappa^3$, where $\bar C>0$ is the constant in Theorem~\ref{thm:GP}.

 We will prove that $\mu_{1}(\kappa H\Fb)<\kappa^2$. We distinguish between three cases:
	\begin{itemize}
		\item[\textbf{Case~1:}]
		Let$H\geq \varepsilon \kappa^3$ and $0<\varepsilon<\bar C$. Multiplying \eqref{eq-2D-GLeq}$_{a}$ by $\psi$ in $L^{2}(\Omega)$ and using \eqref{eq-2D-GLeq}$_{c}$, we observe that,
		\begin{equation}\label{l6est}
			\kappa^{2}\|\psi\|_{L^{4}(\Omega)}^{4}=-\int_{\Omega}\Big(\big|(\nabla-i\kappa H\Ab)\psi\big|^2-\kappa^{2}|\psi|^{2}\Big)\,dx\,.
		\end{equation}
		The Cauchy-Schwarz inequality yields for any $\delta\in(0,\infty)$,
		\begin{equation}
			\int_{\Omega}|(\nabla-i\kappa H\Ab)\psi|^2\,dx\geq (1-\delta)\int_{\Omega}|(\nabla-i\kappa H\Fb)\psi|^{2}\,dx+(1-\delta^{-1})\,(\kappa H)^{2}\int_{\Omega}|(\Ab-\Fb)\psi |^{2}\,dx\,.
		\end{equation}
		This implies that (after using  the min-max principle),
		$$
		\kappa^{2}\|\psi\|_{L^{4}(\Omega)}^{4}\leq -\Big(\mu_{1}(\kappa H \Fb)+\kappa^2\Big)\,\|\psi\|_{L^{2}(\Omega)}^{2}+\delta\, \|(\nabla-i\kappa H\Fb)\psi\|^{2}_{L^{2}(\Omega)}+\delta^{-1}(\kappa H)^{2}\|(\Ab-\Fb)\psi\|_{L^{2}(\Omega)}^{2}\,.
		$$
		Using \eqref{5d-<} and \eqref{8d-<}, we obtain,
		$$
		\kappa^{2}\|\psi\|_{L^{4}(\Omega)}^{4}\leq -\Big(\mu_{1}(\kappa H \Fb)+\kappa^2\Big)\,\|\psi\|_{L^{2}(\Omega)}^{2}+C\,\delta\,\kappa^{2}\,\|\psi\|_{L^{2}(\Omega)}^{2}+C'\,\delta^{-1}\,\kappa^{2}\,\|\psi\|_{L^{2}(\Omega)}^{2}\|\psi\|_{L^{4}(\Omega)}^{4}\,.
		$$
		From \eqref{l3est}, we find that,
		$$\kappa^{2}\|\psi\|_{L^{4}(\Omega)}^{4}\leq -\Big(\mu_{1}(\kappa H \Fb)+\kappa^2\Big)\,\|\psi\|_{L^{2}(\Omega)}^{2}+C\,\delta\,\kappa^{2-\frac{\alpha}{2}}\,\|\psi\|_{L^{4}(\Omega)}^{2}+C'\,\delta^{-1}\,\kappa^{2-\frac{\alpha}{2}}\,\|\psi\|_{L^{4}(\Omega)}^{6}\,.$$
		Remembering that $\|\psi\|_{L^{2}(\Omega)}\neq 0$ and choosing $\delta=\|\psi\|_{L^{4}(\Omega)}^{2}$, we get,
		$$
	\Big(\mu_{1}(\kappa H \Fb)+\kappa^2\Big)\,\|\psi\|_{L^{2}(\Omega)}^{2}\leq\kappa^{2} \|\psi\|_{L^{4}(\Omega)}^{4}\left(C\,\kappa^{-\frac{\alpha}{2}}-1\right)\,.
		$$
		We deduce for $\kappa$ sufficiently large $\mu_{1}(\kappa H\Fb)<\kappa^2$, which what we needed to prove the first case.\\
		\item[\textbf{Case~2:}] Let $c_{0}\kappa^{-1}\leq H < \varepsilon\kappa^{3}$, where $c_{0}>1$ is sufficiently large so that we can use Theorem~\ref{thm.mu1.ful} with $ b=\kappa H $ and $ M=0 $. We observe as $\kappa\to+\infty$,
		\begin{align*}
			\mu_{1}(\kappa H\Fb)&\leq (\kappa H)^\frac{1}{2} |\lambda^{\rm Hess}_{2}(x)|^{1/2} \mathcal{X}_{1}\Big(\varepsilon(x)\Big)+C\,(\kappa H)^{\frac{1 }{4}}\\
			&< \kappa^2 \epsilon^\frac{1}{2} |\lambda^{\rm Hess}_{2}(x)|^{1/2} \mathcal{X}_{1}\Big(\varepsilon(x)\Big)+C\,\kappa \,.
		\end{align*}
		We can choose $\varepsilon>0$ such that $\epsilon^\frac{1}{2} |\lambda^{\rm Hess}_{2}(x)|^{1/2} \mathcal{X}_{1}\Big(\varepsilon(x)\Big)-1<0$ to get,
		$$
	\mu_{1}(\kappa H\Fb)<\kappa^2\,,\qquad{\rm as}~\kappa\longrightarrow+\infty\,.
		$$
		\item[\textbf{Case~3:}] $H<c_{0}\kappa^{-1}$. Let $x_0\in \Omega$ and let $\chi_{\ell}\in C_{c}^{\infty}(\Omega)$ be a cut-off function satisfying:
		\begin{equation}\label{def:chil}
			0\leq \chi_{\ell} \leq 1~{\rm in}~\Omega\,,\quad {\rm supp}\,\chi_{\ell}\subset B(x_{0},\ell)\,,\quad \chi_{\ell}=1~{\rm in}~B\left(x_{0},\frac{\ell}{2}\right)\quad{\rm and}\quad |\nabla\chi_{\ell}|\leq \frac{C}{\ell}\,.
		\end{equation}
	    We may write,
		$$
		\mu_{1}(\kappa H\Fb)\leq \frac{\dsp\int_{\Omega}|(\nabla-i\kappa H\Fb)\chi_{\ell}|^{2}\,dx}{\|\chi_{\ell}\|_{L^{2}(\Omega)}^{2}}\,.
		$$
		Using the assumptions on $\chi_{\ell}$ and the fact that $\Fb\in C^{\infty}(\overline{\Omega})$, a trivial estimate is,
		\begin{align}\label{up:F-varphi3}
			\int_{\Omega}|(\nabla-i\kappa H\Fb)\chi_{\ell}|^{2}\,dx&\leq \int_{B(x_{0},\ell)}|\nabla\chi_{\ell}|^{2}\,dx+\kappa^{2}H^{2}\int_{B(x_{0},\ell)}|\Fb\chi_{\ell}|^{2}\,dx\nonumber\\
			&\leq C\,(1+(\kappa\,H\,\ell)^{2})\,.
		\end{align}
		This implies that
		$$
		\Big(\mu_{1}(\kappa H\Fb)-\kappa^2\Big)\|\chi_{\ell}\|_{L^{2}(\Omega)}^{2}\leq C\Big(1+(\kappa\,H\,\ell)^{2}\Big)-\,(\kappa\ell)^{2}\,.
		$$
		We select $\ell=\kappa^{-\frac{1}{2}}$ and remember that $H<c_{0}\,\kappa^{-1}$.  We find that,
		$$
			\Big(\mu_{1}(\kappa H\Fb)-\kappa^2\Big)\|\chi_{\ell}\|_{L^{2}(\Omega)}^{2}< C\,(1+\kappa^{-1})- \kappa\,,
		$$
		and consequently, for $\kappa$ sufficiently large,
		$$
		\mu_{1}(\kappa H\Fb)<\kappa^2\,.
		$$
	\end{itemize}

\end{proof}

\end{document}